\documentclass[twoside,reqno]{article}
\usepackage{epsfig,cite}
\usepackage{amssymb,amsmath}
\usepackage{times}
\usepackage{epic}
\begin{document}
\sloppy \raggedbottom
\setcounter{page}{1}

\newpage
\setcounter{figure}{0}
\setcounter{equation}{0}
\setcounter{footnote}{0}
\setcounter{table}{0}
\setcounter{section}{0}

\def\beqas{\begingroup \setlength{\arraycolsep}{0.14em} \begin{eqnarray*}}
\def\eeqas{\end{eqnarray*} \endgroup \ignorespaces}

\def\inprod#1{\left\langle #1 \right\rangle}

\newenvironment{proof}{\noindent \textbf{Proof:}}{$\hfill \Box$\\}

\newtheorem{theorem}{Theorem}
\newtheorem{proposition}[theorem]{Proposition}
\newtheorem{lemma}[theorem]{Lemma}
\newtheorem{definition}[theorem]{Definition}

\begin{center}
{\Large
The Hamiltonian $H=xp$ and classification of $\mathfrak{osp}(1|2)$ representations}\\[5mm]
G.~Regniers and J.~Van der Jeugt\\
Department of Applied Mathematics and Computer Science, Ghent University, \\
Krijgslaan 281-S9, B-9000 Gent, Belgium.\\
Gilles.Regniers@UGent.be, Joris.VanderJeugt@UGent.be
\end{center}

\begin{abstract}

The quantization of the simple one-dimensional Hamiltonian $H = xp$ is of interest for its mathematical properties rather than for its physical relevance. In fact, the Berry-Keating conjecture speculates that a proper quantization of $H = xp$ could yield a relation with the Riemann hypothesis. Motivated by this, we study the so-called Wigner quantization of $H = xp$, which relates the problem to representations of the Lie superalgebra $\mathfrak{osp}(1|2)$. In order to know how the relevant operators act in representation spaces of $\mathfrak{osp}(1|2)$, we study all unitary, irreducible $\ast$-representations of this Lie superalgebra. Such a classification has already been made by J.\ W.\ B.\ Hughes, but we reexamine this classification using elementary arguments.

\end{abstract}

\section{Introduction}
The suggestion that the zeros of the Riemann zeta function might be related to the spectrum of a self-adjoint operator $H$ goes back to Hilbert and P\'olya in the early $20$th century. It was not until the works of Selberg~\cite{Selberg-1956} and Montgomery~\cite{Montgomery-1973} that this conjecture gained much credibility. Due to papers by Connes~\cite{Connes-1999} and Berry and Keating~\cite{Berry-1999a, Berry-1999b} in the late 1990s, it appears that the Hilbert-P\'olya conjecture might be related to the classical one-dimensional Hamiltonian $H=xp$. More precisely, Berry and Keating suggest that some sort of quantization of this Hamiltonian might result in a spectrum consisting of the values $t_n$, where the $t_n$ are the heights of the non-trivial Riemann zeros $\frac{1}{2} + it_n$. A proper quantization revealing such a correspondence is, however, not known. \\
These interesting observations stimulated us to perform a different quantization of the Hamiltonian $H=xp$. In Wigner quantization one abandons the canonical commutation relations and instead imposes compatibility between Hamilton's equations and the Heisenberg equations as operator equations. The result is a set of compatibility conditions that are weaker than the canonical commutation relations. This was applied for the first time in a famous paper by Wigner~\cite{Wigner}. \\
Wigner's approach has been applied to many different Hamiltonians, leading to various connections with Lie superalgebras~\cite{LSVdJ-06, LSVdJ-08-1, LVdJ-08}. In the present text, Wigner quantization will lead to the Lie superalgebra $\mathfrak{osp}(1|2)$. Since it is our interest to determine the spectrum of the operators $\hat{H}$ and $\hat{x}$, one needs the action of these operators in representation spaces of $\mathfrak{osp}(1|2)$. We present a classification of all irreducible $\ast$-representations of this Lie superalgebra, thus reconstructing and improving some results by Hughes~\cite{Hughes-1981}. 

\section{Wigner quantization of $H=xp$}
The simplest Hermitian operator that corresponds to our Hamiltonian is given by
\begin{equation}
  \hat{H} = \frac{1}{2} (\hat{x}\hat{p} + \hat{p}\hat{x}).
\end{equation}
Without the assumption of any commutation relations between the position and momentum operators $\hat{x}$ and $\hat{p}$, one can still compute Hamilton's equations
\[
  \dot{\hat{x}} = \frac{\partial \hat{H}}{\partial p} = \hat{x}, \qquad
  \dot{\hat{p}} = - \frac{\partial \hat{H}}{\partial x} = -\hat{p}
\]
and the equations of Heisenberg
\[
  \dot{\hat{x}} = \frac{i}{\hbar} [\hat{H},\hat{x}], \qquad
  \dot{\hat{p}} = \frac{i}{\hbar} [\hat{H},\hat{p}]
\]
and impose that they are equivalent. The resulting compatibility conditions (we choose $\hbar=1$)
\begin{equation} \label{CC_xp}
  [\{\hat{x}, \hat{p}\},\hat{x}] = -2i\hat{x}, \qquad [\{\hat{x}, \hat{p}\},\hat{p}] = 2i\hat{p}
\end{equation}
are weaker than the usual canonical commutation relations $[\hat{x},\hat{p}] = i$. We wish to find self-adjoint operators $\hat{x}$ and $\hat{p}$ such that the compatibility conditions \eqref{CC_xp} are satisfied. For that purpose we define new operators $b^+$ and $b^-$, satisfying $(b^\pm)^\dagger = b^\mp$, as
\[
  b^\pm = \frac{\hat{x} \mp i\hat{p}}{\sqrt{2}}.
\]
One can rewrite the Hamiltonian $\hat{H}$ in terms of the $b^\pm$ as follows:
\[
  \hat{H} = \frac{i}{2} ((b^+)^2 - (b^-)^2).
\]
Evidently the operators $\hat{x}$ and $\hat{p}$ can be expressed as linear combinations of the $b^\pm$. Even the compatibility conditions can be reformulated. They are equivalent to $[\hat{H}, b^\pm] = -i b^\mp$, which in turn can be written as
\begin{equation} \label{osp_def_rel}
  \left[ \{b^-, b^+\}, b^\pm \right] = \pm 2b^\pm.
\end{equation}
These equations are recognized to be the defining relations of the Lie superalgebra $\mathfrak{osp}(1|2)$, generated by the elements $b^+$ and $b^-$. So we have found expressions of all relevant operators in terms of Lie superalgebra generators. \\
A question one might ask is to find the spectrum of $\hat{H}$ and $\hat{x}$ in an $\mathfrak{osp}(1|2)$ representation space, which is only possible once these representation spaces are known. The spectral problem will be tackled in a subsequent paper. Right now, we wish to present a straightforward way of classifying the irreducible $\ast$-representations of $\mathfrak{osp}(1|2)$.

\section{Classification of irreducible $\ast$-representations of $\mathfrak{osp}(1|2)$}
Although we are aware of the classification by Hughes in~\cite{Hughes-1981}, we think it is possible to achieve his results in a more accessible way, based on~\cite{Groenevelt-2004}. In addition we will be able to identify some equivalent representation classes. Before giving the details of our classification, we provide the readers with the necessary definitions and a general outline of how we will construct all irreducible $\ast$-representations of $\mathfrak{osp}(1|2)$.

\subsection{Basic introduction and outline}
We will be dealing with the Lie superalgebra $\mathfrak{osp}(1|2)$, generated by two operators $b^+$ and $b^-$ that are subject to the relations \eqref{osp_def_rel}. The generating operators $b^+$ and $b^-$ are the odd elements of the algebra, while the even elements are 
\[
  h = \frac{1}{2} \{ b^-, b^+ \}, \qquad
  e = \frac{1}{4} \{ b^+, b^+ \}, \qquad
  f = -\frac{1}{4} \{ b^-, b^- \}.
\]
Among others, the following commutation relations can now be computed from the defining relations \eqref{osp_def_rel}:
\[
  [h,e]=2e, \qquad [h,f]=-2f, \qquad [e,f]=h.
\]
One can define a $\ast$-structure on $\mathfrak{osp}(1|2)$, which is an anti-linear anti-multiplicative involution $X \mapsto X^\ast$. For $X, Y \in \mathfrak{osp}(1|2)$ and $a, b \in \mathbb{C}$ we have that $(aX+bY)^\ast = \bar{a}X^\ast + \bar{b}Y^\ast$ and $(XY)^\ast = Y^\ast X^\ast$. Our $\ast$-structure is provided by the dagger operation $X \mapsto X^\dagger$, so we have $\bigl( b^\pm \bigr)^\ast = b^\mp$ and therefore $h^\ast=h, \, e^\ast=-f$ and $f^\ast=-e$. Once we have constructed such a $\ast$-algebra, we need to define representations.
%
\begin{definition}
  Let $\mathcal{A}$ be a $\ast$-algebra, let $\mathcal{H}$ be a Hilbert space and let $\mathcal{D}$ be a dense subspace of $\mathcal{H}$. A $\ast$-representation of $\mathcal{A}$ on $\mathcal{D}$ is a map $\pi$ from $\mathcal{A}$ into the linear operators on $\mathcal{D}$ such that $\pi$ is a representation of $\mathcal{A}$ regarded as a normal algebra, together with the condition
  \begin{equation} \label{star_inprod}
    \inprod{\pi(X)v, w} = \inprod{v, \pi(X^\ast)w}
  \end{equation}
for all $X \in \mathcal{A}$ and $v, w \in \mathcal{D}$. The representation space $\mathcal{D}$, together with the representation $\pi$, is called an $\mathcal{A}$-module. A submodule of $\mathcal{D}$ is a subspace that is closed under the action of $\mathcal{A}$. The representation $\pi$ is said to be irreducible if the $\mathcal{A}$-module $\mathcal{D}$ has no non-trivial submodules.
\end{definition}
The even operators $h$, $e$ and $f$, together with the previously defined $\ast$-structure, form the Lie algebra $\mathfrak{su}(1,1)$. Both $\mathfrak{su}(1,1)$ and $\mathfrak{osp}(1|2)$ possess a Casimir operator, denoted by $\Omega$ and $C$ respectively:
\[
  \Omega = -\frac{1}{4} (4fe + h^2 + 2h), \qquad
  C = -4 \Omega + \frac{1}{2} (b^-b^+ - b^+b^-).
\]
The Casimir elements generate the center of the respective (enveloping) algebras. So $\Omega$ commutes with every element of $\mathfrak{su}(1,1)$ and similary for $C$. Moreover, we have $\Omega^* = \Omega$ and $C^* = C$.

We will construct all possible irreducible $\ast$-representations of $\mathfrak{osp}(1|2)$ starting from one assumption: $h$ has at least one eigenvector in the representation space with eigenvalue $2 \mu$, or
\begin{equation} \label{hv0_2muv0}
  \pi(h) v_0 = 2 \mu \, v_0.
\end{equation}
Starting from this one vector, we will build other basis vectors of the representation space $V$ by letting operators of $\mathfrak{osp}(1|2)$ act on it. After having determined the actions of all $\mathfrak{osp}(1|2)$ operators on all basis vectors of $V$, we will extend the representation $\pi$ to a $\ast$-representation. This is done by defining a sesquilinear form $\inprod{.,.}: V \to \mathbb{C}$, which is to be an inner product that satisfies \eqref{star_inprod}. \\
The stipulation that $\inprod{.,.}$ should be an inner product will be crucial in limiting the possible representation spaces. However, we will postpone the details of this discussion to the point where we have enough arguments for this end. So let us start with the actual construction of the representation space $V$. 

\subsection{Construction of the representation space}
In this section, the $\ast$-structure is of no importance. We will construct an ordinary $\mathfrak{osp}(1|2)$ representation space that we will extend to a $\ast$-representation in the next section. \\
The embedding of $\mathfrak{su}(1,1)$ in $\mathfrak{osp}(1|2)$ implies that any irreducible representation of $\mathfrak{osp}(1|2)$ is a representation of $\mathfrak{su}(1,1)$, the latter being not necessarily irreducible. $V$ can therefore be written as a direct sum of irreducible representation spaces of $\mathfrak{su}(1,1)$, or
\[
  V = \bigoplus_i W_i.
\]
Without loss of generality, we can regard $v_0$ as an element of $W_0$. Since $W_0$ is a representation space of $\mathfrak{su}(1,1)$, we know that
\[
  v_{2k} = \pi(e)^k v_0 \quad \mbox{ and } \quad v_{-2k} = \pi(f)^k v_0
\]
must be elements of $W_0$. All these vectors span the space $W_0$, which is generated by a single vector $v_0$. \\
The action of $b^+$ on any vector of $W_0$ must be a vector outside $W_0$, provided that this action differs from zero. Let us define
\[
  v_1 = \pi(b^+) v_0.
\]
We can say that $v_1$ is an element of $W_1$. Similarly, we can look at the action of $b^-$ on $v_0$:
\[
  v_{-1} = \pi(b^-) v_0.
\]
Since $b^-b^+$ is a diagonal operator (apparent from the definition of the Casimir operator $C$), $\pi(b^-) v_1$ is a certain multiple of $v_0$. At this point however, we cannot be sure that $\pi(b^-) v_1$ is different from zero. Likewise, it is impossible to tell whether $\pi(b^+) v_{-1} \neq 0$. Since we can neither say that $\pi(f) v_1$ is a nonzero multiple of $v_{-1}$, nor that $\pi(e) v_{-1}$ is a multiple of $v_1$, we must regard $v_{-1}$ as an element of a different subspace $W_{-1}$. Note that $W_1$ and $W_{-1}$ are the same spaces when either $\pi(b^-) v_1$ or $\pi(b^+) v_{-1}$ differs from zero. These actions are zero simultaneously only when $\mu=0$. \\
We denote the generating vectors of $W_{-1}$ as $v_{-2k-1} = \pi(f)^k v_{-1}$ and the generating vectors of $W_1$ as $v_{2k+1} = \pi(e)^k v_1$.

\begin{lemma}
The vectors of $W_0$, $W_{-1}$ and $W_1$ are connected by the actions of $b^+$ and $b^-$ in the following manner
\begin{equation}
  v_{2k+1} = \pi(b^+) v_{2k} \qquad \mbox{and} \qquad v_{-2k-1} = \pi(b^-) v_{-2k},
\end{equation}
for every positive integer value of $k$.
\end{lemma}

\begin{proof}
Applying $\pi(b^+)$ to the vector $v_1$ results in a vector of $W_0$ because $\pi(b^+) v_1 = 2\pi(e) v_0$. Thus we find $\pi(b^+) v_2 = \pi(e) v_1 = v_3$. It is clear that this can be generalized to the stated formula for $v_{2k+1}$. The result for $v_{-2k-1}$ can be found analogously.
\end{proof}

Figure \ref{rep1} helps to visualize how the representation space is constructed. We emphasize that the relationship between $v_1$ and $v_{-1}$ is not yet determined.

\vspace{0.7cm}

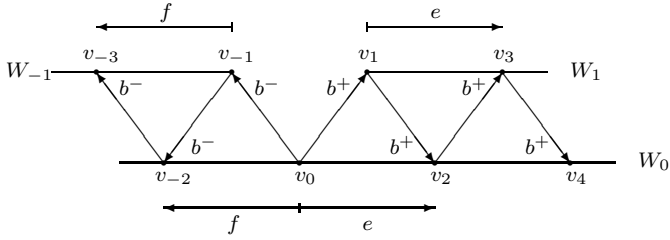
\begin{figure}[ht]
\vspace{0.5cm}
\setlength{\unitlength}{0.60mm}
\begin{center}
  \begin{picture}(15,1)
   \footnotesize{
    \drawline(-55,0)(-15,0)
    \drawline(15,0)(55,0)
    \drawline(-40,-20)(70,-20)
    \put(-65,-1){$W_{-1}$}
    \put(60,-1){$W_1$}
    \put(75,-21){$W_0$}
    \multiput(-30,-20)(30,0){4}{\circle*{1}}
    \put(-32,-24){$v_{-2}$}
    \put(-1,-24){$v_{0}$}
    \put(29,-24){$v_{2}$}
    \put(59,-24){$v_{4}$}
    \multiput(-45,0)(30,0){4}{\circle*{1}}
    \put(-48,3){$v_{-3}$}
    \put(-18,3){$v_{-1}$}
    \put(13,3){$v_{1}$}
    \put(43,3){$v_{3}$}
    \multiput(-30,-20)(30,0){2}{\vector(-3,4){15}}
     \multiput(-40,-5)(30,0){2}{\footnotesize{$b^-$}}
    \multiput(15,0)(30,0){2}{\vector(3,-4){15}}
     \multiput(6,-5)(30,0){2}{\footnotesize{$b^+$}}
    \multiput(0,-20)(30,0){2}{\vector(3,4){15}}
     \multiput(20,-18)(30,0){2}{\footnotesize{$b^+$}}
    \put(-15,0){\vector(-3,-4){15}}
     \put(-24,-18){\footnotesize{$b^-$}}
    \multiput(15,9)(-15,-40){2}{\line(0,1){2}}
    \put(-15,9){\line(0,1){2}}
    \multiput(15,10)(-15,-40){2}{\vector(1,0){30}}
     \multiput(29,12)(-15,-47){2}{$e$}
    \multiput(-15,10)(15,-40){2}{\vector(-1,0){30}}
     \multiput(-31,12)(15,-47){2}{$f$}
   }
  \end{picture}  
\end{center}
\vspace{2.2cm}
\caption{The representation $V$ = $W_{-1} \oplus W_0 \oplus W_1$}
\label{rep1}
\end{figure}
The action of $h$ on the entire representation space $V$ can already be determined.

\begin{lemma}
The action of $h$ on $V$ is given by
\begin{equation}
  \pi(h) v_k = (2\mu+k) v_k,
\end{equation}
for all $k \in \mathbb{Z}$.
\end{lemma}

\begin{proof}
For even values of $k$, this follows just from the relations
\[
  [h,e^k] = 2k e^k, \qquad [h,f^k] = -2k f^k.
\]
For $k=1$, these are the commutation relations $[h,e]=2e$ and $[h,f]=-2f$, and the required identities follow by induction. We then obtain
\[
  \pi(h) \, v_{2k} \, = \, \pi(h) \pi(e)^k \, v_0 \, = \, (2\mu+2k) \, v_{2k}.
\]
For the odd values of $k$, we need $[h,b^\pm] = \pm b^\pm$, which is an instant consequence of equation \eqref{osp_def_rel}. From this, we obtain
\beqas
  \pi(h) \, v_{2k+1} \, = \, \pi(h) \pi(b^+) \, v_{2k} \, = \, (2\mu+2k+1) \, v_{2k+1},
\eeqas
and similarly for $v_{-2k-1}$.
\end{proof}

We would like to determine the actions of $b^+$ and $b^-$ on every vector of $W_0$, $W_{-1}$ and $W_1$. Our method involves defining the action of the Casimir operators on the representation space. We write the respective diagonal actions as
\begin{align*}
  \pi(C) \, v           & = \lambda v      & (\forall \, v \in V), \\
  \pi(\Omega) \, v_{2k} & = -\delta (\delta+1) \, v_{2k} & (\forall \, k \in \mathbb{Z}). 
\end{align*}
We will argue that the choice of $\lambda$ is not independent of $\delta$. It is a nice exercice to show with the help of equation \eqref{osp_def_rel} that
\[
  (b^-b^+ - b^+b^-)^2 = 4 (b^-b^+ - b^+b^-) - 16 \Omega.
\]
This can be used to show that $C^2 = (1-4 \Omega)(2C+4 \Omega)$.
%
%
If we let both sides of this equation act on a vector $v_{2k}$, we get a quadratic equation in $\lambda$. The two possible solutions are
\[
  \lambda_1 = 2 \delta (2 \delta+1) \quad \mbox{ and } \quad
  \lambda_2 = 2 (\delta+1) (2 \delta+1).
\]
We choose $\lambda = \lambda_1$ and remark that the results for the choice $\lambda = \lambda_2$ can be reproduced with the transformation $\delta \to -\delta-1$. \\
In order to be able to determine the actions of $b^+$ and $b^-$ on every vector of $V$, we still need the action of the $\mathfrak{su}(1,1)$ Casimir operator $\Omega$ on $W_{-1}$ and $W_1$.

\begin{lemma} \label{lemma_casimir_odd}
The Casimir operator $\Omega$ acts on $W_{-1}$ and $W_1$ as given by
\begin{equation} \label{casimir_odd}
  \pi(\Omega) v_{2k+1} = - (\delta - \frac{1}{2}) (\delta + \frac{1}{2}) v_{2k+1}, \qquad (k \in \mathbb{Z}).
\end{equation}
As desired, the $\mathfrak{su}(1,1)$-Casimir operator is constant on the subspaces $W_{-1}$ and $W_1$ as well. Moreover, the actions on both subspaces are the same.
\end{lemma}

\begin{proof}
To prove equation \eqref{casimir_odd}, we will calculate $\pi(\Omega) v_{2k+1}$ as $\pi(\Omega b^+) v_{2k}$. From \eqref{osp_def_rel} we can immediately derive that
\[
  [b^-,b^+] b^+ = 2b^+ - b^+ [b^-,b^+].
\]
Using this and twice the definition of the Casimir element $C$, we obtain
\[
  4 \Omega b^+ = b^+ (1 -2C -4 \Omega).
\]
The same formula holds if we change $b^+$ into $b^-$ in both sides of the equation. All of the operators on the right hand side can be applied to vectors of $W_0$. So now $\pi(\Omega b^+) v_{2k}$ can be easily calculated, with equation \eqref{casimir_odd} as a result.
\end{proof}

It has now become straightforward to find the actions of $b^+$ and $b^-$ on all the vectors of $V$.

\begin{proposition} \label{b+b-_delta}
The actions of the operators $b^+$ and $b^-$ on the vectors of $V$ are given by
\begin{equation} \label{b+b-_V}
  \begin{aligned}[c]
    \pi(b^-) v_{2k}    & =  \, (\mu + k + \delta) v_{2k-1},    \\
    \pi(b^-) v_{2k+1}  & =  \, 2 (\mu + k - \delta) v_{2k},    \\
    \pi(b^+) v_{-2k}   & =  \, - (\mu - k - \delta) v_{-2k+1}, \\
    \pi(b^+) v_{-2k-1} & =  \, 2 (\mu - k + \delta) v_{-2k}.
  \end{aligned}
\end{equation}
After the choice $\lambda = \lambda_2$ one would find these actions by means of the transformation $\delta \to -\delta-1$.
\end{proposition}

Since the actions of $h$, $e$ and $f$ follow directly from these relations, we have now constructed all representations of $\mathfrak{osp}(1|2)$ generated by a weight vector $v_0$. It remains to investigate irreducibility and the $\ast$-condition.

\subsection{Extension to $\ast$-representations}
Recall that $V$ is the space spanned by all the vectors $v_k$, $k \in \mathbb{Z}$. We introduce a sesquilinear form $\inprod{.,.}: V \rightarrow \mathbb{C}$ such that
\[
  \inprod{\pi(X)v, w} = \inprod{v, \pi(X^\ast) w}
\]
for all $X \in \mathfrak{osp}(1|2)$ and for all $v,w \in V$. We see that $h^\ast = h$ implies that $\inprod{v_k, v_l} = 0$ for $k \neq l$. This means that the set $\mathcal{S} = \{ v_k | k \in \mathbb{Z}, v_k \neq 0 \}$ forms an orthogonal basis for $V$. We denote by $\mathcal{I}$ the index set such that $v_k \in \mathcal{S}$ for all $k \in \mathcal{I}$. \\
The form $\inprod{.,.}$ is defined by putting
\[
  \inprod{v_k, v_l} = a_k \delta_{kl}, \qquad k,l \in \mathcal{I},
\]
with $a_k$ to be determined and $a_0 = 1$. The definition of a $\ast$-representation requires that the representation space is a Hilbert space, so our sesquilinear form needs to be an inner product. Hence, we want $a_k > 0$ for $k \in \mathcal{I}$. From the action of $h$ and from $h^\ast = h$ we obtain
\[
  2\mu = \inprod{\pi(h)v_0, v_0} = \inprod{v_0, \pi(h)v_0} = 2\bar{\mu},
\]
so $\mu$ must be a real number. Similar calculations for the actions of $\Omega$ and $C$ reveal that both $\delta(\delta+1)$ and $\delta(2\delta+1)$ are real. These two conditions together imply that $\delta$ must be real. \\
From the actions of $b^+$ and $b^-$ and from $(b^\pm)^\ast = b^\mp$, we derive
\[
  a_{2k+1} = \inprod{v_{2k+1}, \pi(b^+)v_{2k}} 
           = \inprod{\pi(b^-)v_{2k+1}, v_{2k}} 
           = 2(\mu + k - \delta) a_{2k}.
\]
In the same way we find
\[
  a_{2k} = \frac{1}{2} (\mu + k + \delta) a_{2k-1}.
\]
Some readers might care for a closed expression for the $a_k$. This is given by
\[
  a_k = \frac{1}{2} \, (3 - (-1)^k) \, (\mu - \delta)_{\left\lceil k/2 \right\rceil}
                                       (\mu + \delta + 1)_{\left\lfloor  k/2 \right\rfloor},
\]
where $(x)_k = x(x+1) \cdots (x+k-1)$ is the classical Pochhammer symbol. \\
We wish to determine under which conditions $\inprod{.,.}$ is an inner product. Alternatively put, for which parameter values is $a_k > 0$ for all $k \in \mathcal{I}$? Starting from $a_0=1$ this can be derived inductively using the two previous equations. We find that all $a_k$ can be positive only if $\mu - \delta > 0$ and $\mu + \delta + 1 > 0$. \\
A similar reasoning should yield a positivity condition for the $a_k$ for negative $k$. However, the resulting conditions $\mu \pm \delta + k > 0$ can never be satisfied for all negative values of $k$. Hence, the representation $\pi$ must have a lowest weight vector, because otherwise it would not be possible to define an inner product on the entire representation space. In this case, the restriction of $\pi$ to an $\mathfrak{su}(1,1)$ subspace is known as a positive discrete series representation.

There are two choices for $\delta$ to obtain a lowest weight representation. One choice is to have $v_0$ as a lowest weight vector, which will arise when $\delta = -\mu$ as one sees from the actions \eqref{b+b-_V}. For $\delta = \mu - 1$ we obtain $\pi(b^+) v_{-2}=0$, in which case $v_{-1}$ is the lowest weight vector.
After one of these choices Proposition \ref{b+b-_delta} must obviously be rewritten. Before we do this, let us make use of the inner product $\inprod{.,.}$ to construct an orthonormal basis $\{ e_k \}$:
\[
  e_{2k} = \frac{v_{2k}}{\|v_{2k}\|}        \quad (k \geq 0), \qquad 
  e_{2k} = (-1)^k \frac{v_{2k}}{\|v_{2k}\|} \quad (k < 0),
\]
and
\[
  e_{2k+1} = \frac{v_{2k+1}}{\|v_{2k+1}\|}            \quad (k \geq 0), \qquad 
  e_{2k+1} = (-1)^{k-1} \frac{v_{2k+1}}{\|v_{2k+1}\|} \quad (k < 0),
\]
%
%
for $k \in \mathcal{I}$. We can now investigate all irreducible $\ast$-representations of $\mathfrak{osp}(1|2)$.

\begin{proposition} \label{prop_two_classes}
The only class of irreducible $\ast$-representations of $\mathfrak{osp}(1|2)$ is a direct sum of two positive discrete series representations of $\mathfrak{su}(1,1)$, determined by a parameter $\mu$. For $0 < \mu \leq \frac{1}{2}$, there is only one irreducible $\ast$-representation of $\mathfrak{osp}(1|2)$. The actions of the generators on the basis vectors $\{ e_k | \, k=0,1,2, \ldots \}$ of the representation space are determined by
\begin{equation} \label{final_actions}
  \begin{array}{rcl}
    \pi(b^+) e_{2k}   & = & \sqrt{2(2\mu+k)} \, e_{2k+1}, \\
    \pi(b^-) e_{2k}   & = & \sqrt{2k}        \, e_{2k-1}, \\
    \pi(b^+) e_{2k+1} & = & \sqrt{2(k+1)}    \, e_{2k+2}, \\
    \pi(b^-) e_{2k+1} & = & \sqrt{2(2\mu+k)} \, e_{2k}. 
  \end{array}
\end{equation}
For $\mu > \frac{1}{2}$, this representation can occur alongside another one, for which the actions of the generators on the basis vectors $\{ e_k | \, k=-1,0,1,2, \ldots \}$ are given by
\begin{equation} \label{equiv_actions}
  \begin{array}{rcl}
    \pi(b^+) e_{2k}   & = & \sqrt{2(k+1)}      \, e_{2k+1}, \\
    \pi(b^-) e_{2k}   & = & \sqrt{2(2\mu+k-1)} \, e_{2k-1}, \\
    \pi(b^+) e_{2k+1} & = & \sqrt{2(2\mu+k)}   \, e_{2k+2}, \\
    \pi(b^-) e_{2k+1} & = & \sqrt{2(k+1)}      \, e_{2k}. 
  \end{array}
\end{equation}
The actions of the other generators follow immediately from these relations and are left for the reader to calculate.
\end{proposition}

\begin{proof}
For $\delta = -\mu$, we get the first representation, which is a lowest weight representation since $\pi(b^-) e_0=0$. It is clear that $\mu$ must be strictly positive so that all the given actions are well defined. The case $\mu=0$ is excluded to be sure that $\pi(b^+) e_{2k}$ differs from zero. \\
In the case of the second representation, for $\delta = \mu - 1$, we must add the condition $\mu > \frac{1}{2}$ to guarantee that $\pi(b^+) e_{-1}$ is well defined and different from zero. We end up with the desired classification.
\end{proof}

Note that if we were to choose $\lambda = \lambda_2$ in the discussion preceding Lemma \ref{lemma_casimir_odd}, we would find exactly the same class of irreducible $\ast$-representations. Indeed, these two representations would pop up for the choices $-\delta-1 = -\mu$ or $-\delta-1 = \mu-1$. It immediately follows that the other actions remain the same in this case. 

Finally, we notice an equivalence between both representation classes in Proposition \ref{prop_two_classes}. Thus, we end up with only one class of irreducible representations of $\mathfrak{osp}(1|2)$.

\begin{theorem}
  The only class of irreducible $\ast$-representations of $\mathfrak{osp}(1|2)$ is a direct sum of two positive discrete series representations of $\mathfrak{su}(1,1)$, determined by a parameter $\mu > 0$. The actions of the generators on the basis vectors $\{ e_k | \, k=0,1,2, \ldots \}$ of the representation space are determined by \eqref{final_actions}.
\end{theorem}

\begin{proof}
  For $\mu > \frac{1}{2}$, define $\bar{e}_k = \bar{e}_{k-1}$ for $k=0,1,2, \ldots$. Then the actions \eqref{equiv_actions} prove to be equivalent to \eqref{final_actions} for $\bar{\mu} = \mu-\frac{1}{2}$. Hence, both representations are equivalent.
\end{proof}

\section{Conclusions and further results}
In this text we have obtained a classification of all irreducible $\ast$-representations of $\mathfrak{osp}(1|2)$. The latter Lie superalgebra showed up naturally in the Wigner quantization of the considered Hamiltonian $H=xp$. Our main concern however, was to investigate the spectrum of the operators $\hat{H}$ and $\hat{x}$. Since these operators are written in terms of generators of $\mathfrak{osp}(1|2)$ we felt the need to explore representations of this Lie superalgebra. They provide us with a suitable framework in which we know how the crucial operators act. 

Results about the spectrum of $\hat{H}$ and $\hat{x}$ have already been found and the details will be published in a subsequent paper, but it is interesting to summarize the results here. \\
In order to find all eigenvalues of one of the operators, one defines a formal eigenvector for a specific eigenvalue $t$,
\[
  v(t) = \sum_{n=0}^\infty \alpha_n(t) e_n,
\]
where the $e_n$ are the eigenvectors of the $\mathfrak{osp}(1|2)$ representation space $V$ and the $\alpha_n(t)$ are unknown coefficients depending on the eigenvalue $t$. Demanding that $v(t)$ is an eigenvector of the operator in question will gives us a three term recurrence relation for the coefficients $\alpha_n(t)$. These coefficients are then identified with the orthogonal polynomials that comply with the same recurrence relation. The spectrum of the operator is then equal to the support of the weight function of this type of orthogonal polynomials. \\
Concretely we have that the spectrum of $\hat{H}$ is related to Meixner-Pollaczek polynomials and is equal to $\mathbb{R}$ with multiplicity two. Generalized Hermite polynomials are connected with the spectrum of $\hat{x}$, which is simply $\mathbb{R}$. \\
Recall that Wigner quantization is a somewhat more general approach than canonical quantization. This means that one should be able to recover the canonical case from the results after Wigner quantization. Indeed, our results prove to be compatible with the well-known canonical case for the representation parameter $\mu=\frac{1}{4}$.

\section*{Acknowledgments}
G.~Regniers was supported by project P6/02 of the Interuniversity Attraction Poles Programme (Belgian State --- Belgian Science Policy)


\end{document}